\documentclass[11pt,a4paper,oneside,english]{article}

\usepackage{array}

\usepackage{romannum}

 \usepackage{multirow}
\usepackage{babel}
\usepackage[utf8]{inputenc}
\usepackage[T1]{fontenc} 
\usepackage{amsmath}
\usepackage{amsthm}
\usepackage{amssymb}
\usepackage{physics}
\usepackage{graphicx}
\usepackage{booktabs}
\usepackage[hidelinks]{hyperref}
\usepackage{makecell}
\usepackage{romannum}
\usepackage{longtable}
\usepackage[flushleft]{threeparttable}
\usepackage{fullpage,authblk}
\usepackage{float}
\usepackage{blindtext}
\usepackage{lmodern}
\usepackage[numbers]{natbib}
\numberwithin{equation}{section}

\newtheorem{theorem}{Theorem}
\providecommand{\keywords}[1]{
  \small	
  \textbf{\textit{Keywords---}} #1
}

\title{\textbf{A Three-Party Lightweight Quantum Key Distribution Protocol in a
Restricted Quantum Environment}}

\begin{document}

\author[1]{Mustapha Anis Younes \footnote{Corresponding Author: \texttt{mustaphaanis.younes@univ-bejaia.dz}}}
\author[2]{Sofia Zebboudj \footnote{\texttt{sofia.zebboudj@univ-ubs.fr}}}
\author[3]{Abdelhakim Gharbi \footnote{\texttt{abdelhakim.gharbi@univ-bejaia.dz}}}
\affil[1,2]{Université de Bejaia, Faculté des Sciences Exactes, Laboratoire de Physique Théorique, 06000 Bejaia, Algérie}
\affil[2]{ENSIBS, Université de Bretagne Sud, 56000 Vannes, France}
\date{\vspace{-8ex}}

\maketitle 
\pagenumbering{arabic}
\pagenumbering{arabic}

\begin{abstract}

 This study proposes a new lightweight quantum key distribution (LQKD) protocol based on the four-particle cluster state within a quantum-restricted environment. The protocol enables a quantum-capable user to simultaneously establish two separate secret keys with two "classical" users, who are limited to performing only the Hadamard operation and measurements in the $Z$ basis. By adopting a one-way qubit transmission approach, the proposed protocol addresses several limitations of existing semi-quantum key distribution (SQKD) schemes that rely on two-way or circular transmission methods: (1) it eliminates the need for classical participants to be equipped with costly quantum devices to defend against quantum Trojan horse attacks; (2) it reduces the qubit transmission distance; and (3) it achieves higher qubit efficiency. Consequently, the proposed LQKD protocol is both more lightweight and practical than existing SQKD protocols. Furthermore, the security analysis shows that, in the ideal case, the protocol achieves the same level of security as fully quantum protocols. Finally, the study proves the unconditional security of the protocol in the non-ideal case, demonstrating a noise tolerance close to that of the BB84 protocol.
    
\end{abstract}


\section{Introduction}

Quantum key distribution (QKD) protocols allow two parties to establish a shared secret key. By leveraging the laws of quantum mechanics, the key remains secure even against an all-powerful adversary limited only by the laws of nature. The first QKD protocol \cite{Bennett2014}, known as BB84, was introduced by Bennett and Brassard in 1984 and relies on the no-cloning theorem of non-orthogonal single-photon states. Since then, numerous QKD protocols have been proposed \cite{Ekert1991,Bennett1992,BechmannPasquinucci1999,Scarani2004,Grosshans2002,Lo2012,Mayers, Long2002,Shu2023,Sharma2021,Hwang2003,Lai2020}
, utilizing various types of quantum resources. Although QKD can achieve unconditional security \cite{Shor2000,Lucamarini2009,Lo2004,Barrett2012,Leverrier2009}, most existing protocols require participants to possess full quantum capabilities, which is unrealistic as not all users can afford or operate expensive quantum devices.

To address this issue, Boyer et al. \cite{Boyer2007,Boyer2009} introduced the concept of "semi-quantum environment", which includes two types of users: quantum and classical. According to the definition, quantum users possess full quantum capabilities, whereas classical users are restricted to performing the following operations: (1) reflect particles without disturbance; (2) measure qubits in the $Z$ basis ${\ket{0}, \ket{1}}$; (3) generate qubits in the $Z$ basis; and (4) reorder qubits via different delay lines. In 2007, Boyer et al. \cite{Boyer2007} proposed the first semi-quantum key distribution (SQKD) protocol. Since then, various SQKD protocols have been developed \cite{Boyer2017,He2018,Zou2015,Iqbal2020,Li2016,Ye2018,Hajji2021}. In 2015, Krawec \cite{Krawec2015} introduced the mediated framework, which involves an untrusted, fully quantum third party (TP) acting as a mediator to help two classical participants to establish a secure key, further reducing the quantum burden on the users. Several SQKD protocols have since been designed based on this framework \cite{Liu2018,Lin2019,Lu2020,Chen2021,Guskind2022,Ye2023,Tsai2019,Tsai2021,Hwang2023}. Krawec was also the first to prove the unconditional security of an SQKD protocol \cite{Krawec2015a}. Later, SQKD protocols based on single-qubit states were shown to be unconditionally secure as well \cite{Krawec2014, Krawec2016,Krawec2016a,Zhang2020}, demonstrating that semi-quantum protocols can achieve the same level of security as fully quantum ones.

The generalization of quantum key distribution (QKD) to multipartite scenarios is useful for many applications \cite{Murta2020,Wehner2018} (e.g. broadcasting) and achieving this task using pairwise QKD is an inefficient approach \cite{Epping2017}. The first multi-user semi-quantum key distribution (MSQKD) protocol was introduced by Zhang et al. \cite{XianZhou2009} in 2009, using single-qubit states. However, their scheme suffers from low qubit efficiency. In 2019, Zhou et al. \cite{Zhou2019} used the properties of the four-particle cluster state to propose a new MSQKD protocol. Although their protocol addressed the efficiency issue, it exhibited relatively low noise tolerance compared to standard pairwise SQKD protocols. In 2022, Tian et al. \cite{Tian2022} proposed an MSQKD protocol based on Hyperentangled Bell states, which improved the channel capacity. The following year, Ye et al. \cite{Ye2023} introduced a mediated MSQKD protocol, in which classical users could establish a secret key while remaining secure against a potentially dishonest quantum user. This protocol employs Bell states and demonstrates a noise tolerance close to that of the BB84 protocol.

Even though the semi-quantum framework was shown to be very promising in terms of practicality. Typical SQKD protocols (multi-party or not) adopt a two-way or circular qubit transmission approach. This gives rise to several challenges, namely: 

\begin{itemize}
    \item[(1)] Classical users need to be equipped with additional costly quantum devices to mitigate the quantum Trojan horse attacks. This not only increases the quantum burden on the classical users, which goes against the original intent of the semi-quantum environment, but it also reduces the qubit efficiency. 
    \item[(2)] The qubits travel longer distances than in one-way schemes. This leads to higher qubit loss and gives to the attacker more opportunities to interact with the qubits.  
    \item[(3)] The efficiency of SQKD protocols is significantly lower than fully-quantum protocols.
\end{itemize}

To address the first two issues, Tsai et al. \cite{Tsai2019,Tsai2021,Tsai2022,Tsai2023} proposed a new semi-quantum environment known as the "restricted quantum environment", which also includes two types of participants: classical and quantum. Quantum users still possess full quantum capabilities, whereas classical users are limited to the following two quantum capabilities: (i) perform the Hadamard operation; and (ii) measure qubits in the $Z$ basis.  In this environment, classical participants are free from qubit-preparation device and they do not need additional costly quantum devices to mitigate against the Trojan horse attack, significantly reducing their quantum burden. In 2019, Tsai et al. \cite{Tsai2019} proposed the first lightweight mediated semi-quantum key distribution (LMSQKD) protocol based on single photons. However, the efficiency of this protocol remained relatively low. Two years later, Tsai and Yang \cite{Tsai2021} introduced another LMSQKD protocol, this time based on Bell states, which achieved significantly better qubit efficiency and allowed two classical participants to share a secret key with the help of a dishonest third party. However, those protocols only allow two users to share keys and do not consider the multiparty scenario. In 2024, Tsai and Wang \cite{Tsai2024} introduced the first mediated" multiparty" quantum key distribution (M-MQKD) protocol in the restricted quantum environment, based on measurement properties of quantum graph states. This protocol enables a classical user to share secret keys with multiple remote classical users in a general network environment. Nevertheless, due to the use of quantum repeaters, the qubit efficiency of this protocol remains relatively low compared to existing multi-party SQKD protocols.

In this paper, we propose a novel lightweight quantum key distribution (LQKD) protocol based on the four-particle cluster state. This state possesses a high persistency of entanglement \cite{Briegel2001}, making it a more advantageous quantum resource for multi-party schemes than other classes of maximally entangled states. The proposed LQKD protocol is a three-party scheme that enables quantum Charlie to simultaneously establish two secret keys: one with classical Alice and one with classical Bob, gaining time and qubit efficiency compared to two-party schemes. Since it adopts the restricted quantum environment proposed by Tsai et al. \cite{Tsai2024}, the protocol is inherently immune to quantum Trojan horse attacks and reduces the transmission cost of qubits thanks to its one-way qubit transmission approach. This not only enhances its practicality compared to other multi-party SQKD schemes, but also significantly reduces the quantum burden on classical users, thus aligning with the original intent of the semi-quantum model. Furthermore, the proposed protocol achieves higher qubit efficiency than other multi-party schemes in the three-party scenario. The security analysis also shows that the protocol is secure against both internal and external attackers employing well-known attacks such as intercept-and-resend, measure-and-resend, and collective attacks. Finally, we prove the unconditional security of the protocol by deriving a lower bound on the secret key rate in the asymptotic scenario. Our results show that the protocol achieves a noise tolerance of $9.68\%$, which is close to the $11\%$ threshold \cite{Shor2000} of the BB84 protocol.

The remainder of this paper is organized as follows: Section \ref{sec2} provides a brief overview of the mathematical tools used throughout the paper. Section \ref{sec3_protocol} presents the specific steps of the proposed three-party LQKD protocol. Section \ref{sec4_security} describes the detailed security analysis of the protocol. A discussion and comparison with other schemes are given in Section \ref{sec5_comp}. Finally, Section \ref{sec6_conc} concludes the paper

\section{Preliminaries}\label{sec2}

In this section, we offer a brief overview of the mathematical tools that we use throughout this paper.

\subsection{The Hadamard operation}
We begin by introducing the Hadamard operation $H$ that is represented by 

\begin{align} \label{hadamard}
        H &= \frac{1}{\sqrt{2}}(\ketbra{0}{0} + \ketbra{1}{0} + \ketbra{0}{1} + \ketbra{1}{1}) \\
          &= \cfrac{1}{\sqrt{2}}\begin{pmatrix}
                         1 & 1 \\
                         1 & -1
                 \end{pmatrix}.
    \end{align}

\subsection{The Bell states}
Next, we present the Bell states, which we define as

\begin{align} \label{phi+}  
    \ket{\phi^+} &= \frac{1}{\sqrt{2}} ( \ket{00} + \ket{11} ), \\ \label{phi-}  
    \ket{\phi^-} &= \frac{1}{\sqrt{2}} ( \ket{00} - \ket{11} ), \\ \label{psi+}  
    \ket{\psi^+} &= \frac{1}{\sqrt{2}} ( \ket{01} + \ket{10} ), \\ \label{psi-} 
    \ket{\psi^-} &= \frac{1}{\sqrt{2}} ( \ket{01} - \ket{10} ).
\end{align}

The Bell states form an orthonormal basis in the four-dimensional Hilbert space $\mathcal{H}^2\otimes \mathcal{H}^2$, known as the Bell basis.

In this paper, we use the propriety between Bell states and the Hadamard operation to detect eavesdroppers. The Hadamard operation affects the Bell states as follows:

\begin{align} \label{h on phi+}
    (H \otimes H) \ket{\phi^+} &= \ket{\phi^+}, \\ \label{h on phi-}
    (H \otimes H) \ket{\phi^-} &= \ket{\psi^+}, \\ \label{h on psi+}
    (H \otimes H) \ket{\psi^+} &= \ket{\phi^-}, \\  \label{h on psi-}
    (H \otimes H) \ket{\psi^-} &= -\ket{\psi^-}. 
\end{align}

\subsection{The four-particle cluster state}
Continuing our overview, we now introduce the four-particle cluster state as it is used in the proposed protocol as information carrier. It can be expressed in the following ways:

\begin{align}\label{cluster}
    \ket{C}_{1234} &= \frac{1}{2}\big( \ket{0000} + \ket{0110} + \ket{1001} - \ket{1111} \big)_{1234}, \\  \label{cluster_bell_bell}
                      &= \frac{1}{2}\Big( \ket{\phi^+}_{12}\ket{\phi^-}_{34} + \ket{\phi^-}_{12}\ket{\phi^+}_{34} + \ket{\psi^+}_{12}\ket{\psi^+}_{34} - \ket{\psi^-}_{12}\ket{\psi^-}_{34} \Big)_{1234}, \\  
     \label{cluster_bell14}
                      &= \frac{1}{\sqrt{2}} \Big(\ket{0}_{2}\ket{\phi^+}_{14}\ket{0}_{3} + \ket{1}_{2}\ket{\phi^-}_{14}\ket{1}_{3}\Big). \\
    \label{cluster_bell23}
                      &= \frac{1}{\sqrt{2}} \Big(\ket{0}_{1}\ket{\phi^+}_{23}\ket{0}_{4} + \ket{1}_{1}\ket{\phi^-}_{23}\ket{1}_{4}\Big). 
\end{align}

The four-particle cluster state is a maximally entangled state. It was shown in \cite{Briegel2001,Duer2004} that in the case of $n > 3$, $n$-particle cluster states have two basic properties: (1) maximum connectedness; and (2) a high persistency of entanglement. These properties are easily observed in Eq. (\ref{cluster_bell23}). For instance, measuring qubits 1 and 4 in the $Z$ basis projects qubits 2 and 3 into a pure Bell state with certainty. Also, it can be seen that if we perform a local measurement on one or two qubits of the cluster state, the remaining particles will still be entangled. These properties make cluster states a more advantageous resource for multi-party schemes compared to other classes of maximally entangled states, such as GHZ states, which are more vulnerable to decoherence.

\subsection{Entropy functions}
Finally, we introduce some elements of information theory as they are going to be useful to us in the computation of the secret key rate. Suppose we have a probability distribution $p_1, p_2, \cdots, p_n$, such as $\sum_{i=1}^{n} p_{i}=1$. The \textit{Shannon entropy} associated with this probability distribution is defined as 

\begin{align} 
    H(p_1, p_2, \cdots, p_n) = - \sum_{i=1}^{n} p_{i}\log p_{i},
\end{align}
where all logarithms in this paper are base two unless specified otherwise. Additionally, we denote $h(x) = -p\log p - (1-p)\log (1-p)$, as the \textit{binary Shannon entropy} which is a special case of $H(p_1, p_2, \cdots, p_n)$ where we only have two probabilities $(p, 1-p)$.

The generalization of Shannon entropy for a quantum state $\rho$ is given by the \textit{von Neumann entropy} 

\begin{align}
    S(\rho) = -tr(\rho \log \rho).
\end{align}

Let $\rho$ be a finite dimensional density operator with eigenvalues $\{\lambda_1, \lambda_2, \cdots, \lambda_{n}\}$. Then, the Von Neumann entropy can be re-expressed as

\begin{align}
    S(\rho) = - \sum_{i=1}^{n} \lambda_{i}\log \lambda_{i}.
\end{align}

Consider now a bipartite state $\rho_{AB} \in \mathcal{H}_A \otimes \mathcal{H}_B$ that represent the density operator of the composite system $AB$. The state of the subsystems $A$ and $B$ are given by the partial trace operation $\rho_A = tr_B(\rho_{AB})$ and $\rho_B = tr_A(\rho_{AB})$. We often write $S(AB)$ to mean $S(\rho_{AB})$ and $S(A)$ to mean $S(\rho_{A})$. We can now define \textit{the conditional von Neumann entropy} as 

\begin{align}
    S(A|B) = S(AB) - S(B).
\end{align}

\section{The proposed protocol}\label{sec3_protocol}

In this paper, we propose a three-party lightweight quantum key distribution protocol. The protocol allows a quantum-capable user, Charlie, to simultaneously establish two different secret keys with two classical users, Alice and Bob (one shared with Alice and the other with Bob). The classical users only have two capabilities: 

\begin{enumerate}
    \item Measure qubits in the $Z$ basis.
    \item Perform the Hadamard operation $H$.
\end{enumerate}

As for Charlie, she is capable to perform any quantum operation as well as accessing a quantum memory. Between Charlie and each of the participants there is an insecure quantum channel as well as an authenticated classical channel.

\subsection{Steps of the protocol}

The steps of our protocol unfolds as follow:

\subsubsection{Charlie's preparation}

Charlie prepares $N = 4n(1+\epsilon)$ four-particle cluster states, where $n$ is an integer and $\epsilon$ is a fixed positive parameter. Charlie divides the states into four sequences 

\begin{align}
    S_l = \{s_l^1, s_l^2, \cdots, s_l^N\},
\end{align}
where $l = 1, 2, 3, 4$, and the $l-th$ particle of each cluster state constitute the sequence $S_l$.

Charlie keeps the sequences $S_3$ and $S_4$ in her quantum memory and sends the sequence $S_1$ and $S_2$ to Alice and Bob, respectively.

\subsubsection{Alice's and Bob's operations}

Upon receiving their respective sequences, Alice and Bob randomly and independently apply either the identity operation $I$ (i.e. do nothing) or the Hadamard operation $H$ to each qubit in their sequences, before proceeding to measure them in the $Z$ basis. They each record their measurement outcomes in their respective classical registers. Note that Alice and Bob register the measurement outcome as $0$ if they obtain the state $\ket{0}$ and $1$ otherwise.

\subsubsection{Alice's and Bob's public announcement}

Alice and Bob announce, via the authenticated classical channels, which operation they implemented on each qubit before their measurements.

\subsubsection{Charlie's operations}

Depending on Alice's and Bob's chosen operation, Charlie performs the following actions:

    \begin{itemize}
        \item \textbf{Case 01:} Alice and Bob both implements $I$ on the qubits $s_1^i$ and $s_2^i$, respectively. In this case, Charlie will also implement the same operation on $s_3^i$ and $s_4^i$ before measuring them in $Z$ basis. The state (\ref{cluster}) of the system collapse into one of the four following outcomes with equal probability:

        \begin{align} \label{outcome 1st case}
        \ket{C}_{1234} \stackrel{Measure}{\longrightarrow} \begin{cases}
        \ket{00}_{14}\ket{00}_{23},  \\
        \ket{00}_{14}\ket{11}_{23},  \\
        \ket{11}_{14}\ket{00}_{23},  \\
        \ket{11}_{14}\ket{11}_{23}. 
    \end{cases}
      \end{align}
       According to the results shown above, Charlie's measurement outcomes on the qubits $s_3^i$ must always match Bob's measurement outcomes on $s_2^i$, while her outcomes on $s_4^i$ must always align with Alice's outcome on $s_1^i$. 
       
       To detect eavesdroppers, Charlie randomly selects half of the positions corresponding to this case and requests Alice and Bob to disclose their measurement results for those positions over the public channel. She then verifies if the measurement outcomes are consistent with the expected results shown in the above equations. If the observed error rate exceeds a predefined threshold, the protocol is aborted; Otherwise, Alice (Bob) uses her (his) remaining results to contribute to the formation of her (his) shared raw key $R_{CA}$ ($R_{CB}$) with Charlie.

        \item \textbf{Case 02:} Alice implements $I$ on $s_1^i$, while Bob applies $H$ on $s_2^i$. In this scenario, Charlie will perform the same operation as Alice on $s_4^i$ and the same operation as Bob on $s_3^i$ before measuring them in the $Z$ basis. Based on the relation between Bell states and the Hadamard operation given in Eqs. (\ref{h on phi+}) and (\ref{h on phi-}), the system (\ref{cluster_bell23}) evolves as follows:

         \begin{align}
            (I \otimes H \otimes H \otimes I) \ket{C}_{1234} &= \frac{1}{\sqrt{2}} \Big(\ket{0}_{1}\ket{\phi^+}_{23}\ket{0}_{4} + \ket{1}_{1}\ket{\psi^+}_{23}\ket{1}_{4}\Big), \\
            &=  \frac{1}{2} \begin{pmatrix}
                \ket{0}_{1}\ket{00}_{23}\ket{0}_{4} + \ket{0}_{1}\ket{11}_{23}\ket{0}_{4} \\
                \ket{1}_{1}\ket{01}_{23}\ket{1}_{4} + \ket{1}_{1}\ket{10}_{23}\ket{1}_{4}
            \end{pmatrix}
        \end{align}
        Upon measurement, the above state collapse into one the following results with equal probability

        \begin{align} \label{outcome 2nd case}
        (I \otimes H \otimes H \otimes I) \ket{C}_{1234} \stackrel{Measure}{\longrightarrow} \begin{cases}
        \ket{00}_{14}\ket{00}_{23},  \\
        \ket{00}_{14}\ket{11}_{23},  \\
        \ket{11}_{14}\ket{01}_{23},  \\
        \ket{11}_{14}\ket{10}_{23}. 
    \end{cases}
      \end{align}
        
        From these outcomes, we observe that Charlie’s results on $s_4^i$ must always match Alice’s outcomes on $s_1^i$, while her results on $s_3^i$ are expected to be correlated with Bob’s outcomes on $s_2^i$ if Alice measured the state $\ket{0}_1$, and anti-correlated if she measured $\ket{1}_1$.

        To detect a potential eavesdropper, Charlie randomly selects half of the positions corresponding to this case and requests Alice and Bob to disclose their outcomes for those positions. She then analyzes the error rate and if it exceeds a predefined threshold, the protocol is aborted; otherwise, Alice’s remaining outcomes contribute to the formation of her shared raw key $R_{CA}$ with Charlie.

        \item \textbf{Case 03:} Alice applies $H$ and Bob $I$ on $s_1^i$ and $s_2^i$, respectively. In this case, Charlie will apply the same operation as Bob on $s_3^i$ and the same operation as Alice on $s_4^i$. According to Eqs.(\ref{h on phi+}) and (\ref{h on phi-}), the state (\ref{cluster_bell14}) evolves in the following manner:

        \begin{align}
            (H \otimes I \otimes I \otimes H) \ket{C}_{1234} &= \frac{1}{\sqrt{2}} \Big(\ket{0}_{2}\ket{\phi^+}_{14}\ket{0}_{3} + \ket{1}_{2}\ket{\psi^+}_{14}\ket{1}_{3}\Big), \\
            &=  \frac{1}{2} \begin{pmatrix}
                \ket{0}_{2}\ket{00}_{14}\ket{0}_{3} + \ket{0}_{2}\ket{11}_{14}\ket{0}_{3} \\
                + \ket{1}_{2}\ket{01}_{14}\ket{1}_{3} + \ket{1}_{2}\ket{10}_{14}\ket{1}_{3}
            \end{pmatrix}.
        \end{align}

        After measurement, the participants obtain one of the following outcomes with equal probability 

        \begin{align} \label{outcome 3rd case}
        (H \otimes I \otimes I \otimes H) \ket{C}_{1234} \stackrel{Measure}{\longrightarrow} \begin{cases}
        \ket{00}_{14}\ket{00}_{23},  \\
        \ket{11}_{14}\ket{00}_{23},  \\
        \ket{01}_{14}\ket{11}_{23},  \\
        \ket{10}_{14}\ket{11}_{23}. 
    \end{cases}
      \end{align}
        Based on these results, Charlie's measurement outcomes on the qubits $s_3^i$ must always align with Bob's measurement outcomes. As for the qubits $s_4^i$, Charlie expects her measurement results to be correlated with Alice's results if Bob measured the state $\ket{0}_2$ and anti-correlated otherwise.
        
        To check for eavesdropping, Charlie randomly selects half of the positions corresponding to this case and requests Alice and Bob to reveal their measurement results for those selected positions. If the observed error rate exceeds a predefined threshold, the protocol is aborted; otherwise, Bob uses his remaining results to contribute to the formation of his shared raw key $R_{CB}$ with Charlie.

        \item \textbf{Case 04:} Both Alice and Bob implement $H$. In this situation, Charlie performs a Bell measurement on $s_3^i$ and $s_4^i$. Based on Eqs. (\ref{h on phi+}-\ref{h on psi-}), the state of the system can be written as follows:

        \begin{align}
                (H \otimes H \otimes I \otimes I) \ket{C}_{1234} 
                &= \frac{1}{2} \begin{bmatrix}
                     \ket{00}_{12} \left(\cfrac{\ket{\phi^-}_{34} + \ket{\psi^+}_{34}}{\sqrt{2}}\right) \\
                     + \ket{11}_{12}\left(\cfrac{\ket{\phi^-}_{34} - \ket{\psi^+}_{34}}{\sqrt{2}}\right) \\
                     + \ket{01}_{12}\left(\cfrac{\ket{\phi^+}_{34} + \ket{\psi^-}_{34}}{\sqrt{2}}\right) \\
                     + \ket{10}_{12}\left(\cfrac{\ket{\phi^+}_{34} - \ket{\psi^-}_{34}}{\sqrt{2}}\right)
                \end{bmatrix}
                \end{align} 

       The system then collapses into one of the eight following outcomes with equal probability:

       \begin{align} \label{outcome 4th case}
        (H \otimes H \otimes I \otimes I) \ket{\Phi}_{1234}  \stackrel{Measure}{\longrightarrow} \begin{cases}
        \ket{00}_{12}\ket{\phi^-}_{34}, \quad \ket{00}_{12}\ket{\psi^+}_{34}  \\
        \ket{11}_{12}\ket{\phi^-}_{34}, \quad \ket{11}_{12}\ket{\psi^+}_{34}  \\
        \ket{01}_{12}\ket{\phi^+}_{34}, \quad \ket{01}_{12}\ket{\psi^-}_{34}  \\
        \ket{10}_{12}\ket{\phi^+}_{34}, \quad \ket{10}_{12}\ket{\psi^-}_{34}. 
    \end{cases}
      \end{align}

        Based on these results, if Charlie obtains the state $\ket{\phi^-}_{34}$ or $\ket{\psi^+}_{34}$, she expects Alice and Bob to announce correlated results. On the other hand, if she measures $\ket{\phi^+}_{34}$ or $\ket{\psi^-}_{34}$, she expects Alice and Bob to announce anti-correlated results. 
        
        Charlie requests Alice and Bob to disclose all their measurement results associated with this case and evaluates the error rate; if it exceeds a predetermined threshold, the protocol aborts.

    \end{itemize}

It is important to note that Charlie never reveals the outcomes of her measurements. The expected results for each of the four emerging cases are shown in Table 1.

\subsubsection{Classical post-processing}

Charlie performs error correction (EC) and privacy amplification (PA) separately with Alice and Bob. Based on their respective raw keys, $R_{CA}$ and $R_{CB}$, she extracts two final secret keys, denoted $K_{AC}$ and $K_{BC}$, which she shares with Alice and Bob, respectively.

\begin{table}[h]
\small
\caption{\label{table 1}The expected results based on the users' operations.}
\begin{tabular}{cccc}
\toprule
Alice's operation & Bob's operation & Charlie's operation & Expected results \\ 
\midrule

$I$ & $I$ & $I \otimes I$ & 
$\begin{array}{l}
mr_A = mr_{C_4} \\
mr_B = mr_{C_3}
\end{array}$ \\

\midrule

$I$ & $H$ & $H \otimes I$ & 
$\begin{array}{l}
mr_A = mr_{C_4} = 
\begin{cases}
0 \Rightarrow mr_B = mr_{C_3} \\
1 \Rightarrow mr_B = mr_{C_3} \oplus 1
\end{cases}
\end{array}$ \\

\midrule

$H$ & $I$ & $I \otimes H$ & 
$\begin{array}{l}
mr_B = mr_{C_3} = 
\begin{cases}
0 \Rightarrow mr_A = mr_{C_4} \\
1 \Rightarrow mr_A = mr_{C_4} \oplus 1
\end{cases}
\end{array}$ \\

\midrule

$H$ & $H$ & Bell measurement & 
$\begin{array}{l}
mr_{C_{34}} = \ket{\phi^-} \text{ or } \ket{\psi^+} \Rightarrow mr_A = mr_B \\
mr_{C_{34}} = \ket{\phi^+} \text{ or } \ket{\psi^-} \Rightarrow mr_A = mr_B \oplus 1
\end{array}$ \\

\bottomrule
\end{tabular}
\noindent\footnotesize
$mr_A$ denotes Alice's measurement result on $s_1^i$, \\
$mr_B$ denotes Bob's measurement result on $s_2^i$, \\
$mr_{C_j}$ denotes Charlie's measurement result on $s_j^i$, where $j = 3, 4$.
\end{table}

\section{Security analysis}\label{sec4_security}

In this section, we analyze the security of the proposed scheme against both internal and external attackers. We demonstrate in both scenarios that the protocol is secure against well-knows attacks such as the intercept and resend attack, the measure and resend attack, the Trojan horse attacks, and collective attack. We also prove the unconditional security of the proposed protocol by computing a lower bound for the secret key rate $r$ in the asymptotic scenario.

\subsection{Internal attack}

Dishonest participants represent a more serious threat than external attackers because they actively participate in the protocol. Therefore, internal attacks require closer attention in the security analysis. In the proposed scheme, Alice and Bob play symmetric roles. Therefore, without loss of generality, we consider Bob to be the dishonest participant; his goal is to steal information about Alice's and Charlie's secret key without being detected. In order to do so, he may launch any possible attack on the quantum channel between Alice and Charlie. Since the protocol employs a one-way communication structure, Bob can only launch his attacks once when the qubits are transmitted from Charlie to Alice.

\subsubsection{Intercept and resend attack}

In this attack, Bob first intercepts the sequence $S_1$ that Charlie sends to Alice. He stores it in his quantum memory and sends a previously prepared fake particle sequence to Alice. Bob then waits for Alice to publicly announce the operations she applied to each qubit. Once he obtains this information, he applies the same operations to the qubits of $S_1$. However, this strategy is bound to fail. Suppose Bob prepares his fake sequence in the $Z$ basis. Since this preparation is random, there is always a nonzero probability that his measurement outcomes differ from the initial states of the fake sequence, and therefore from the outcomes announced by Alice.

More specifically, for each qubit, Bob has a probability of $\frac{1}{2}$ of obtaining a result that matches Alice’s outcome, regardless of the case. Consequently, the total probability that Bob is detected using this eavesdropping strategy is given by $1 - \left(\frac{1}{2}\right)^M$, where $M$ denotes the number of positions used for the eavesdropping check. This probability approaches 1 as $M$ becomes sufficiently large.

\subsection{Measure and resend attack}

Bob intercepts the sequence $S_1$, measures each qubit in the $Z$ basis, and then forwards the measured sequence to Alice. If Alice applies the identity to $s_1^i$, Bob successfully learns one bit of the raw key shared between Alice and Charlie without being detected. However, if Alice applies the Hadamard operation, his interference will be detected in both the following cases:

\begin{itemize}
    \item If Bob directly measures $s_1^i$ and $s_2^i$ in the $Z$ basis, the state $\ket{C}_{1234}$ collapses into one of the outcomes shown in Eq. (\ref{outcome 1st case}). However, when Alice and Charlie later apply $H$ to $s_1^i$ and $s_4^i$, respectively, there is only a $\frac{1}{2}$ probability that their outcomes will be consistent with the expected results shown in Eq. (\ref{outcome 3rd case}).
    
    \item If Bob applies $H$ on $s_2^i$, the state $\ket{C}_{1234}$ collapses into one of the outcomes given in Eq. (\ref{outcome 2nd case}). However, when Alice later applies $H$ to $s_1^i$, there is a $\frac{1}{2}$ probability that Charlie obtains the correct measurement outcome when she measures $s_3^i$ and $s_4^i$ in the Bell basis.

\end{itemize}
Overall, the total probability for Bob to be detected using this strategy is $1 - \left(\frac{3}{4}\right)^M$, where $M$ denotes the number of positions used for the eavesdropping check. This probability approaches 1 when $M$ is large enough.


In the proposed scheme, when we examine the possible outcomes listed in Table \ref{table 1}, it is clear that Bob cannot deduce Alice’s outcomes from his own. Therefore, a second strategy for Bob would be to measure the qubits $s_1^i$ and $s_2^i$ in the Bell basis. While this might allow him to correlate Alice's measurement results with his own when they select the same operation, it would simultaneously disrupt the intended correlations between their qubits and those held by Charlie. As a result the probability for Bob to be detected using this strategy is $1 - \left(\frac{3}{4}\right)^M$, where $M$ denotes the number of positions used for the eavesdropping check. This probability converges to 1 as $M$ becomes sufficiently large.

\subsection{Trojan horse attacks}

Quantum Trojan horse attacks, as described in \cite{Deng2005,Cai2006}, are common implementation-dependent attacks in which an attacker can exploit invisible photons, or delayed photons to extract confidential information from users without being detected. Regardless of the resources employed, the strategy remains the same: the attacker intercepts the qubits sent to the users and attaches a probing photon to each one. These probing photons are prepared and inserted in such a way that they go undetected by the users' detectors. When a participant performs an operation on a received qubit, the probing photon undergoes the same operation. After retrieving the Trojan horse photons, the attacker can obtain confidential information about the users' operations by measuring them.

For a Trojan horse attack to succeed, the attacker must retrieve the malicious photons they previously inserted.  Typical semi-quantum protocols rely on two-way or circular communication structures, which inherently offer the attacker an opportunity to recover the malicious photons they previously introduced. Consequently, such protocols are considered vulnerable to Trojan horse attacks. To defend against these threats, semi-quantum protocols often equip classical participants with costly quantum devices such as photon number splitters (PNS) and wavelength filters (WF). However, this approach may contradict the original purpose of using the semi-quantum environment as the burden on the classical participant is significantly increased.  Furthermore, the detectors of quantum Trojan horse attacks must use additional qubits to verify the attacks, so the qubit efficiency of quantum protocols also drops. In contrast, the proposed scheme adopts a one-way qubit transmission method(i.e., qubits are not returned by the Alice and Bob). This structure prevents Bob from retrieving the previously inserted malicious photons, thereby eliminating any opportunity to acquire information about the Alice’s secret. As a result, the protocol demonstrates robustness against Trojan horse attacks without requiring the participants to rely on costly quantum devices to mitigate these attacks.

\subsection{Collective attack}

In general, the most effective strategy for Bob to gain information about Alice’s and Charlie’s secret key without being detected is to perform a unitary attack. Depending on how Bob interacts with the quantum signals sent from Charlie to Alice, and when or how he performs his measurements, these attacks can be categorized into three classes: individual, collective, and coherent attacks. Individual attacks impose the greatest constraints on Bob, while coherent attacks are the most powerful, where Bob is only restricted by the fundamental laws of quantum physics. It has been shown in \cite{Christandl2009,Renner2007} that proving security against collective attacks implies security against coherent attacks. Collective attacks are defined as follows:

\begin{itemize}
    \item Bob attacks each traveling qubit sent by Charlie to Alice through the quantum channel independently of the others, using the same strategy.
    \item Bob entangles an ancillary particle with each qubit sent by Charlie to Alice.
    \item Bob can keep his ancilla in his quantum memory until any later time convenient to him. He can then perform the optimal measurement based on the side information obtained through public discussion.
\end{itemize}

\begin{theorem}

    Assume Bob performs a collective attack on the qubits sent by Charlie to Alice by applying a unitary operation $U_B$ to the transmitted qubits along with their associated ancillary systems. In this case, there exists no unitary operation that enables Bob to extract any meaningful information about the secret key shared between Alice and Charlie without disturbing the original quantum systems, resulting in a nonzero probability of being detected during the eavesdropping check.  
\end{theorem}

\begin{proof}

Bob intercepts $S_1$ and applies $U_B$ to each qubit $s_1^i$ along with an associated ancilla $B$, initially prepared in an arbitrary normalized state $\ket{e}$. The action of $U_B$ is defined as follows:

\begin{align} \label{U_B on 0}
    U_{B} (\ket{0} \otimes \ket{e}) &= a \ket{0}\ket{e_{00}} + b\ket{1} \ket{e_{01}}, \\ \label{U_B on 1}
    U_{B} (\ket{1} \otimes \ket{e}) &= c \ket{0}\ket{e_{10}} + d \ket{1}\ket{e_{11}},
\end{align}
such that $\abs{a}^{2}+\abs{b}^{2}=\abs{c}^{2}+\abs{d}^{2}=1$ and the states $\ket{e_{ij}}$ are orthogonal states that Bob can distinguish. The state of the quantum system can be written in the following way:

\begin{align}
    U_{B} \ket{C}_{1234} \ket{e}
    &= \frac{1}{2} \begin{bmatrix}
        \big( U_{E} \ket{0}_1 \ket{e} \big) \big( \ket{000}_{234} + \ket{110}_{234} \big)  \\
        + \big( U_{E} \ket{1}_1 \ket{e} \big) \big( \ket{001}_{234} - \ket{111}_{234} \big) 
    \end{bmatrix} \nonumber \\
    &= \frac{1}{2} \begin{bmatrix}
        \ket{0}_1 \otimes \begin{pmatrix}
            a \ket{e_{00}} \big( \ket{000}_{234} + \ket{110}_{234} \big) \\
            + c \ket{e_{10}} \big( \ket{001}_{234} - \ket{111}_{234} \big) 
        \end{pmatrix} \\
        + \ket{1}_1 \otimes \begin{pmatrix}
            b \ket{e_{01}}  \big( \ket{000}_{234} + \ket{110}_{234} \big) \\
            + d \ket{e_{11}} \big( \ket{001}_{234} - \ket{111}_{234} \big) 
        \end{pmatrix}
    \end{bmatrix} 
\end{align}

This is the state that Alice, Bob, and Charlie share among themselves. Let us now examine how this state evolves depending on the different cases that emerge in the protocol:

\begin{itemize}
    \item \textbf{Case 01:} The operations $I\otimes I\otimes I\otimes I$ are applied on the qubits $s_1^i$, $s_2^i$, $s_3^i$, and $s_4^i$ respectively. In this case, the state of the system remains the same and can be rewritten as:

     \begin{align}
        \ket{C_1}&= \frac{1}{2} \begin{bmatrix}
            a \ket{e_{00}} \otimes
            \big( \ket{00}_{14}\ket{00}_{23} + \ket{00}_{14}\ket{11}_{23} \big) \\
            + c \ket{e_{10}} \otimes \big( \ket{01}_{14}\ket{00}_{23} - \ket{01}_{14}\ket{11}_{23}\big) \\
            +  
            b \ket{e_{01}} \otimes \big( \ket{10}_{14}\ket{00}_{23} + \ket{10}_{14}\ket{11}_{23} \big)  \\
            + d \ket{e_{11}} \otimes \big( \ket{11}_{14}\ket{00}_{23} - \ket{11}_{14}\ket{11}_{23} \big) 
        \end{bmatrix}
    \end{align}

    \item \textbf{Case 02:} The operations $I\otimes H\otimes H\otimes I$ are performed on the qubits $s_1^i$, $s_2^i$, $s_3^i$, and $s_4^i$ respectively. In this case, the state of the system becomes:
    \begin{align}
        \ket{C_2}&= \frac{1}{2} \begin{bmatrix}
            a \ket{e_{00}} \otimes
            \big( \ket{00}_{14}\ket{00}_{23} + \ket{00}_{14}\ket{11}_{23} \big) \\
            + c \ket{e_{10}} \otimes \big( \ket{01}_{14}\ket{01}_{23} - \ket{01}_{14}\ket{10}_{23} \big) \\
            +  
            b \ket{e_{01}} \otimes \big( \ket{10}_{14}\ket{00}_{23} + \ket{10}_{14}\ket{11}_{23} \big) \\
            + d \ket{e_{11}} \otimes \big( \ket{11}_{14}\ket{01}_{23} - \ket{11}_{14}\ket{10}_{23} \big) 
        \end{bmatrix}  
    \end{align}

    \item \textbf{Case 03:} The operations $H\otimes I\otimes I\otimes H$ are performed on the qubits $s_1^i$, $s_2^i$, $s_3^i$, and $s_4^i$ respectively. In this case, the state of the system evolves as follows:

    \begin{align}
        \ket{C_3}&= \frac{1}{4} \begin{bmatrix}
            \big( \ket{00}_{14}\ket{00}_{23} + \ket{01}_{14}\ket{11}_{23} \big) \otimes \ket{f_{00}} \\
            + \big( \ket{01}_{14}\ket{00}_{23} + \ket{00}_{14}\ket{11}_{23} \big) \otimes \ket{f_{01}} \\
            + \big( \ket{10}_{14}\ket{00}_{23} + \ket{11}_{14}\ket{11}_{23} \big) \otimes \ket{f_{10}} \\
            + \big( \ket{11}_{14}\ket{00}_{23} + \ket{10}_{14}\ket{11}_{23} \big)\otimes \ket{f_{11}}
        \end{bmatrix} 
    \end{align}
    where

    \begin{align} \label{f_00}
        \ket{f_{00}} &= a\ket{e_{00}} + b\ket{e_{01}} + c\ket{e_{10}} + d\ket{e_{11}}, \\ \label{f_01}
        \ket{f_{01}} &= a\ket{e_{00}} + b\ket{e_{01}} - c\ket{e_{10}} - d\ket{e_{11}}, \\ \label{f_10}
        \ket{f_{10}} &= a\ket{e_{00}} - b\ket{e_{01}} + c\ket{e_{10}} - d\ket{e_{11}}, \\ \label{f_11}
        \ket{f_{11}} &= a\ket{e_{00}} - b\ket{e_{01}} - c\ket{e_{10}} + d\ket{e_{11}},
    \end{align}
    and $\braket{f_{00}}{f_{01}} = \braket{f_{10}}{f_{11}} = 0$.

    \item \textbf{Case 04:} The operations $H\otimes H\otimes I\otimes I$ are applied on the qubits $s_1^i$, $s_2^i$, $s_3^i$, and $s_4^i$ respectively. In this case, the state of the system can be written as follows:


    \begin{align}
        \ket{C_4}&= \frac{1}{4\sqrt{2}} \Big[ 
         \ket{00}_{12}\big( a\ket{e_{00}}+b\ket{e_{01}} \big) \big( \ket{\phi^+}_{34}+\ket{\phi^-}_{34}+\ket{\psi^+}_{34}-\ket{\psi^-}_{34} \big)            
     \nonumber \\
    &\quad\quad -\ket{00}_{12}\big( c\ket{e_{10}}+d\ket{e_{11}} \big) \big( \ket{\phi^+}_{34}-\ket{\phi^-}_{34}-\ket{\psi^+}_{34}-\ket{\psi^-}_{34} \big)    
     \nonumber \\
    &\quad\quad +\ket{01}_{12}\big( a\ket{e_{00}}+b\ket{e_{01}} \big) \big( \ket{\phi^+}_{34}+\ket{\phi^-}_{34}-\ket{\psi^+}_{34}+\ket{\psi^-}_{34} \big)    
    \nonumber \\
    &\quad\quad +\ket{01}_{12}\big( c\ket{e_{10}}+d\ket{e_{11}} \big) \big( \ket{\phi^+}_{34}-\ket{\phi^-}_{34}+\ket{\psi^+}_{34}+\ket{\psi^-}_{34} \big)    
    \nonumber \\
    &\quad\quad +\ket{10}_{12}\big( a\ket{e_{00}}-b\ket{e_{01}} \big) \big( \ket{\phi^+}_{34}+\ket{\phi^-}_{34}+\ket{\psi^+}_{34}-\ket{\psi^-}_{34} \big)    
    \nonumber \\
    &\quad\quad -\ket{10}_{12}\big( c\ket{e_{10}}-d\ket{e_{11}} \big) \big( \ket{\phi^+}_{34}-\ket{\phi^-}_{34}-\ket{\psi^+}_{34}-\ket{\psi^-}_{34} \big)    
    \nonumber \\
    &\quad\quad +\ket{11}_{12}\big( a\ket{e_{00}}-b\ket{e_{01}} \big) \big( \ket{\phi^+}_{34}+\ket{\phi^-}_{34}-\ket{\psi^+}_{34}+\ket{\psi^-}_{34} \big)    
    \nonumber \\
    &\quad\quad +\ket{11}_{12}\big( c\ket{e_{10}}-d\ket{e_{11}} \big) \big( \ket{\phi^+}_{34}-\ket{\phi^-}_{34}+\ket{\psi^+}_{34}+\ket{\psi^-}_{34} \big)    
    \Big].
    \end{align}

\end{itemize}

In the proposed protocol, Charlie requests Alice and Bob to reveal their measurement outcomes at specific positions via the public authenticated channel and compares them with her own to verify whether they match the expected results presented in Table \ref{table 1}. Therefore, to avoid detection, Bob must carefully adjust the unitary operation $U_B$ accordingly.

In order to avoid detection in the first case, Bob must ensure two things: (1) his measurement results on $s_2^i$ must always be correlated with Charlie’s results on $s_3^i$; and (2) Alice’s outcomes on $s_1^i$ must always align with Charlie’s outcomes on $s_4^i$. As for the second case, Bob must prevent the results of Alice and Charlie on qubits $s_1^i$ and $s_4^i$ from being anti-correlated. In order to achieve that, Bob must set:

\begin{align} \label{b=c=0}
    b=c=0,
\end{align}

Now to pass detection in the third case, Alice and Charlie outcomes on $s_1^1$ and $s_4^i$ must be correlated when Bob's outcome is $0$ and anti-correlated otherwise. To satisfy that, Bob must set the incorrect terms $\ket{f_{01}}$ and $\ket{f_{10}}$ as a zero vector. This requires meeting the condition:

\begin{align} \label{condition2}
    a\ket{e_{00}} + b\ket{e_{01}} - c\ket{e_{10}} - d\ket{e_{11}} &= 0, \\
    a\ket{e_{00}} - b\ket{e_{01}} + c\ket{e_{10}} - d\ket{e_{11}} &= 0.
\end{align}

Finally, to pass the security check in the fourth case, Bob has to make sure that his measurements are correlated with Alice's when Charlie measures the state $\ket{\phi^-}_{23}$ or $\ket{\psi^+}$ and anti-correlated otherwise. To achieve that, Bob must fulfill both conditions (\ref{b=c=0}) and (\ref{condition2}).

Regardless of whether Bob first adjusts $U_B$ to the first or second condition, it will lead to same following relation:

\begin{align}
    a\ket{e_{00}} = d\ket{e_{11}}, 
\end{align}
which means that Bob cannot distinguish $\ket{e_{00}}$ from $\ket{e_{11}}$ unless he sets $a=d=0$. However, combining this with (\ref{b=c=0}), it contradicts the equation $\abs{a}^{2}+\abs{b}^{2}=\abs{c}^{2}+\abs{d}^{2}=1$. Therefore, in order to satisfy both this equation and the condition $b = c = 0$, Bob must set $a = d = 1$. In this scenario, Bob cannot distinguish $\ket{e_{00}}$ from $\ket{e_{11}}$. Thus, the definition of $U_B$ can be rewritten as:

\begin{align} \label{U_E on 0}
    U_{E} (\ket{0} \otimes \ket{e}) &= \ket{0}\ket{e_{00}}, \\ \label{U_E on 1}
    U_{E} (\ket{1} \otimes \ket{e}) &= \ket{1}\ket{e_{00}},
\end{align}

We can see that Bob’s ancilla is independent of the qubits being attacked. Therefore, even though Charlie cannot detect the attack, Bob cannot extract any information about Charlie's and Alice's shared raw key. Thus, we have demonstrated that there exists no unitary operation that allows Bob to obtain any useful information about the secret key without being detected.

\end{proof}

\subsection{External attack}

In this scenario, we consider an external eavesdropper, Eve, who simultaneously attacks both quantum channels. Her objective is to obtain, without being detected, the secret keys that Charlie shares with each participant. Since it has been shown in \cite{Christandl2009,Renner2007} that security against collective attacks is sufficient to demonstrate security against any arbitrary general attacks, we focus solely on proving security against collective attacks in the following subsection. 

\begin{theorem}

    Assume Eve performs a collective attack on the qubits sent by Charlie to the participants by applying a unitary operation $U_E$ to the transmitted qubits along with their associated ancillary systems. In this case, there exists no unitary operation that enables Eve to extract any useful information about the secret keys shared between Charlie and each of the participants without disturbing the original quantum systems, resulting in a nonzero probability of being detected during the eavesdropping check.  
\end{theorem}

\begin{proof}

Eve intercepts the sequences $S_1$ and $S_2$ and applies her unitary operation $U_E$ to each pair of qubits $s_1^i$ and $s_2^i$ along with their associated ancillary particle $E$, initially prepared in the state $\ket{e^\prime}$. We define the action of $U_E$ as follows:

\begin{align}
     U_{E} (\ket{00} \otimes \ket{e^\prime}) &= a_0 \ket{00}\ket{e_{00}} + a_1 \ket{01} \ket{e_{01}} + a_2 \ket{10} \ket{e_{01}} + a_3 \ket{11} \ket{e_{11}}, \\ 
     U_{E} (\ket{01} \otimes \ket{e^\prime}) &= b_0 \ket{00}\ket{f_{00}} + b_1 \ket{01} \ket{f_{01}} + b_2 \ket{10} \ket{f_{01}} + b_3 \ket{11} \ket{f_{11}}, \\
     U_{E} (\ket{10} \otimes \ket{e^\prime}) &= c_0 \ket{00}\ket{g_{00}} + c_1 \ket{01} \ket{g_{01}} + c_2 \ket{10} \ket{g_{01}} + c_3 \ket{11} \ket{g_{11}}, \\
     U_{E} (\ket{11} \otimes \ket{e^\prime}) &= d_0 \ket{00}\ket{h_{00}} + d_1 \ket{01} \ket{h_{01}} + d_2 \ket{10} \ket{h_{01}} + d_3 \ket{11} \ket{h_{11}}.
\end{align}
Where the states $\ket{e_{ij}}$, $\ket{f_{ij}}$, $\ket{g_{ij}}$, and $\ket{h_{ij}}$, such as $ij \in \{00, 01, 10, 11\}$, are sixteen states that Eve can distinguish. As for the coefficients $a_k$, $b_k$, $c_k$, and $d_k$, such as $k = 0, 1, 2, 3$, they satisfy the following relations:

\begin{align} \label{norm a_i}
    \abs{a_0}^2 + \abs{a_1}^2 + \abs{a_2}^2 + \abs{a_3}^2 &= 1, \\ \label{norm b_i}
    \abs{b_0}^2 + \abs{b_1}^2 + \abs{b_2}^2 + \abs{b_3}^2 &= 1, \\ \label{norm c_i}
    \abs{c_0}^2 + \abs{c_1}^2 + \abs{c_2}^2 + \abs{c_3}^2 &= 1, \\ \label{norm d_i}
    \abs{d_0}^2 + \abs{d_1}^2 + \abs{d_2}^2 + \abs{d_3}^2 &= 1.
\end{align}

After Eve's operation, the quantum system evolves as follows:

\begin{align}
    U_E \ket{C}_{1234}\ket{e^\prime} &= \frac{1}{2} \begin{bmatrix}
        \big( U_{E} \ket{00}_{12} \ket{e^\prime} \big)\ket{00}_{34} 
        + \big( U_{E} \ket{01}_{12} \ket{e^\prime} \big)\ket{10}_{34}  \\
        +\big( U_{E} \ket{10}_{12} \ket{e^\prime} \big)\ket{01}_{34} 
        - \big( U_{E} \ket{11}_{12} \ket{e^\prime} \big)\ket{11}_{34}
    \end{bmatrix}
     \nonumber \\
    &= \frac{1}{2} \begin{bmatrix}
        \big( 
            a_0 \ket{00}_{12}\ket{e_{00}} +  a_1 \ket{01}_{12}\ket{e_{01}} + a_2 \ket{10}_{12}\ket{e_{10}} + a_3 \ket{11}_{12}\ket{e_{11}} \big) \otimes \ket{00}_{34} \\
            + \big( 
            b_0 \ket{00}_{12}\ket{f_{00}} +  b_1 \ket{01}_{12}\ket{f_{01}} + b_2 \ket{10}_{12}\ket{f_{10}} + b_3 \ket{11}_{12}\ket{f_{11}} \big) \otimes \ket{10}_{34} \\
            + \big( 
            c_0 \ket{00}_{12}\ket{g_{00}} +  c_1 \ket{01}_{12}\ket{g_{01}} + c_2 \ket{10}_{12}\ket{g_{10}} + c_3 \ket{11}_{12}\ket{g_{11}} \big) \otimes \ket{01}_{34} \\
            - \big( 
            d_0 \ket{00}_{12}\ket{h_{00}} +  d_1 \ket{01}_{12}\ket{h_{01}} + d_2 \ket{10}_{12}\ket{h_{10}} + d_3 \ket{11}_{12}\ket{h_{11}} \big) \otimes \ket{11}_{34}
    \end{bmatrix}
\end{align}

If Eve wants to escape detection in the case where both Alice and Bob apply the identity operation $I$, she must set the following coefficients to zero:
\begin{align} \label{a_1 =...=d_2= 0}
a_1 = a_2 = a_3 = b_0 = b_2 = b_3 = c_0 = c_1 = c_3 = d_0 = d_1 = d_2 = 0.
\end{align}

Once this condition is satisfied, the state described above can evolve in the following ways, depending on the operations performed by the participants:

\begin{itemize}
    \item \textbf{Case 01:} Both Alice and Bob apply the identity $I$. The state of the system can be written 

    \begin{align}
        \ket{C^\prime_1} &= \frac{1}{2} \begin{pmatrix}
            a_0 \ket{0000}_{1234}\ket{e_{00}} + b_1 \ket{0110}_{1234}\ket{f_{01}}  \\
            + c_2 \ket{1001}_{1234}\ket{g_{10}} + d_3 \ket{0110}_{1234}\ket{h_{11}}
        \end{pmatrix}
    \end{align}
    \item \textbf{Case 02:} Alice applies the identity while Bob applies Hadamard. The state of the system evolves as follows:

    \begin{align} \label{c_2_prime}
        \ket{C^\prime_2} &= \frac{1}{4} \begin{bmatrix}
            \big( \ket{00}_{14}\ket{00}_{23} + \ket{00}_{14}\ket{11}_{23} \big)\big( a_0 \ket{e_{00}} + b_1 \ket{f_{01}} \big)  \\
         + \big( \ket{00}_{14}\ket{01}_{23} + \ket{00}_{14}\ket{10}_{23} \big)\big( a_0 \ket{e_{00}} - b_1 \ket{f_{01}} \big)  \\
         + \big( \ket{10}_{14}\ket{00}_{23} + \ket{11}_{14}\ket{11}_{23} \big)\big( c_2 \ket{g_{10}} - d_3 \ket{h_{11}} \big)  \\
         + \big( \ket{11}_{14}\ket{01}_{23} + \ket{11}_{14}\ket{10}_{23} \big)\big( c_2 \ket{g_{10}} + d_3 \ket{h_{11}} \big) 
        \end{bmatrix}
    \end{align}

    \item \textbf{Case 03:} Alice applies the Hadamard operation and Bob applies the identity. In this case, the state of the system becomes:

    \begin{align} \label{c_3_prime}
        \ket{C^\prime_3} &= \frac{1}{4} \begin{bmatrix}
            \big( \ket{00}_{14}\ket{00}_{23} + \ket{11}_{14}\ket{00}_{23} \big)\big( a_0 \ket{e_{00}} + c_2 \ket{g_{10}} \big) \\
        + \big( \ket{01}_{14}\ket{00}_{23} + \ket{10}_{14}\ket{00}_{23} \big)\big( a_0 \ket{e_{00}} - c_2 \ket{g_{10}} \big)  \\
        + \big( \ket{00}_{14}\ket{11}_{23} + \ket{11}_{14}\ket{11}_{23} \big)\big( b_1 \ket{f_{01}} - d_3 \ket{h_{11}} \big)  \\
        + \big( \ket{01}_{14}\ket{11}_{23} + \ket{10}_{14}\ket{11}_{23} \big)\big( c_2 \ket{g_{10}} + d_3 \ket{h_{11}} \big) 
        \end{bmatrix}.
    \end{align}

    \item \textbf{Case 04:} Both Alice and Bob perform the Hadamard operation. In this case, we can write the state of the system as follows:

    \begin{align} \label{c_4_prime}
        \ket{C^\prime_4} &= \frac{1}{4\sqrt{2}} \begin{bmatrix}
            \Big( \big(\ket{00}_{12} + \ket{11}_{12}\big)\ket{\phi^+}_{34} + \big( \ket{01}_{12} + \ket{11}_{12} \big)\ket{\phi^-}_{34}\Big)\Big(a_0 \ket{e_{00}} - d_3 \ket{h_{00}}\Big)  \\
        + \Big( \big(\ket{00}_{12} + \ket{11}_{12}\big)\ket{\phi^+}_{34} + \big( \ket{01}_{12} + \ket{11}_{12} \big)\ket{\phi^-}_{34}\Big)\Big(a_0 \ket{e_{00}} + d_3 \ket{h_{00}}\Big) \\
        + \Big( \big(\ket{00}_{12} - \ket{11}_{12}\big)\ket{\psi^+}_{34} + \big( \ket{01}_{12} - \ket{11}_{12} \big)\ket{\psi^-}_{34}\Big)\Big(b_1 \ket{f_{01}} + c_2 \ket{g_{10}}\Big)  \\
        + \Big( \big(\ket{11}_{12} - \ket{00}_{12}\big)\ket{\psi^-}_{34} + \big( \ket{01}_{12} + \ket{10}_{12} \big)\ket{\psi^+}_{34}\Big)\Big(b_1 \ket{f_{01}} - c_2 \ket{g_{10}}\Big)
        \end{bmatrix}.
    \end{align}

\end{itemize}


As mentioned before, Charlie asks Alice and Bob to disclose their measurements results at specific positions via the public authenticated channels. She then compares them with her own and see if they are consistent with the expected results shown in Table \ref{table 1}. In order for Eve to avoid detection in the last three cases, she must set the incorrect terms in Eqs. (\ref{c_2_prime}–\ref{c_4_prime}) to zero vectors. More precisely, for each of the last three cases, Eve must satisfy the following conditions, respectively:

\begin{align}
    a_0 \ket{e_{00}} &= b_1 \ket{f_{01}}, \\
    c_2 \ket{g_{10}} &= d_3 \ket{h_{11}}.
\end{align}

\begin{align}
     a_0 \ket{e_{00}} &= c_2 \ket{g_{10}}, \\
     b_1 \ket{f_{01}} &= d_3 \ket{h_{11}}.
\end{align}

\begin{align}
     a_0 \ket{e_{00}} &= d_3 \ket{h_{11}}, \\
     b_1 \ket{f_{01}} &= c_2 \ket{g_{10}}.
\end{align}
Combining those conditions results in:

\begin{align}
    a_0 \ket{e_{00}} = b_1 \ket{f_{01}} = c_2 \ket{g_{10}} = d_3 \ket{h_{11}}.
\end{align}
This relation contradicts the distinguishability of the states $\ket{e_{00}}$, $\ket{f_{01}}$, $\ket{g_{10}}$, and $\ket{h_{11}}$, unless Eve sets all coefficients $a_0$, $b_1$, $c_2$, and $d_3$ to zero. However, combining that with the first condition (\ref{a_1 =...=d_2= 0}) leads to a contradiction with Eqs. (\ref{norm a_i}-\ref{norm d_i}).

To satisfy Eqs. (\ref{norm a_i}-\ref{norm d_i}), Eve must set $a_0 = b_1 = c_2 = d_3 = 1$. As a result, she can no longer distinguish between the states $\ket{e_{00}}$, $\ket{f_{01}}$, $\ket{g_{10}}$, and $\ket{h_{11}}$. In the case where Eve remains undetected, her unitary operation $U_E$ acts as follows:

\begin{align}
    U_{E} (\ket{00} \otimes \ket{e^\prime}) &= \ket{00} \ket{e_{00}}, \\ 
    U_{E} (\ket{01} \otimes \ket{e^\prime}) &= \ket{01} \ket{e_{00}}, \\
    U_{E} (\ket{10} \otimes \ket{e^\prime}) &= \ket{10} \ket{e_{00}}, \\
    U_{E} (\ket{11} \otimes \ket{e^\prime}) &= \ket{11} \ket{e_{00}}.
\end{align}

We observe that Eve’s ancilla becomes independent of the qubits being attacked. Consequently, she cannot extract any information about the participants' secret keys by measuring her ancillary particles. On the contrary, if Eve attempts to make the ancillary states distinguishable, her interference will be detected with nonzero probability, as the participants’ measurement results would then deviate from those expected in Table \ref{table 1}.
    
\end{proof}

\subsection{Key rate}

In this subsection, we evaluate a bound for the key rate in the asymptotic scenario (i.e., Charlie prepares an infinite number of four-particle cluster states). To achieve that, we use the proof method proposed by Krawec in \cite{Krawec2016a}. Since the key distribution process between Charlie and each of the participants is identical, we only focus on analyzing the key rate between Charlie and Alice. The key rate is defined as follows:

\begin{align}
    r = \lim_{n\rightarrow{+\infty}} \cfrac{l(n)}{n},
\end{align}
where $n$ in the length of the raw key and $l(n)$ is the length of final secret key after Charlie and Alice performed error correction and privacy amplification. If $r > 0$, then Charlie and Alice can distill a secure secret key.

Assuming collective attacks, it was shown in \cite{Renner2007,Devetak2005} that the key rate can be expressed as follows: 

\begin{align}
    r \geq \inf\big(S(A|E) - H(A|C)\big).
\end{align}

The computation of $H(A|C)$ is usually trivial and the real challenge is to compute $S(A|E)$. In the remaining of the subsection, we assume all quantum channels to be ideal and that all the observed noise comes from Eve's attack.

\subsubsection{Attack strategy}

Let $V_E$ be the unitary operation that models Eve's attack and $\ket{\chi}$ be some arbitrary, normalized state that represents the initial state of Eve's ancilla. The action the $V_E$ is defined as follows:

\begin{align}
    V_E \ket{0}\ket{\chi} &= \ket{0}\ket{e_0} + \ket{1}\ket{e_1}, \\
    V_E \ket{1}\ket{\chi} &= \ket{0}\ket{e_2} + \ket{1}\ket{e_3},
\end{align}
where the states $\ket{e_i}$ satisfy the following relation:

\begin{align}
    \braket{e_0} + \braket{e_1} = \braket{e_2} + \braket{e_3} &= 1, \\
    \braket{e_0}{e_2} = \braket{e_1}{e_3} &= 0
\end{align}

Eve launches her attack on each qubit sent by Charlie to Alice. After her attack, the system becomes:

\begin{align} \label{v_e on C}
    V_E \ket{C}_{1234} \ket{\chi} =  \frac{1}{2} \begin{bmatrix}
                         \ket{0}_1 \begin{pmatrix}
                             \ket{e_0}(\ket{000}_{234} + \ket{110}_{234}) \\
                             \ket{e_2}(\ket{001}_{234} - \ket{111}_{234})
                         \end{pmatrix} \\
                         \ket{1}_1 \begin{pmatrix}
                             \ket{e_1}(\ket{000}_{234} + \ket{110}_{234}) \\
                             \ket{e_3}(\ket{001}_{234} - \ket{111}_{234})
                         \end{pmatrix}
                 \end{bmatrix}.
\end{align}

We now construct the density operator describing the quantum system shared between the users and Eve, conditioning on the cases that are used to construct the raw key $R_{CA}$ (i.e., when Alice apply the identity operation on her qubit). Therefore we have the following two situations:

\begin{enumerate}
    \item \textbf{Situation 01:} Both Alice and Bob apply the identity operation. Knowing that Charlie also apply the same operation on her qubits. The state (\ref{v_e on C}) can then be written as follows:

    \begin{align}
     \frac{1}{2} \begin{bmatrix}
                         ( \ket{00}_{14}\ket{00}_{23} + \ket{00}_{14}\ket{11}_{23})\otimes\ket{e_0} \\
                         +( \ket{01}_{14}\ket{00}_{23} - \ket{01}_{14}\ket{11}_{23})\otimes\ket{e_2} \\
                         +( \ket{10}_{14}\ket{00}_{23} + \ket{10}_{14}\ket{11}_{23})\otimes\ket{e_1} \\
                         +( \ket{11}_{14}\ket{00}_{23} - \ket{11}_{14}\ket{11}_{23})\otimes\ket{e_3}
                 \end{bmatrix}.
\end{align}

Upon measurement, the system can be describe by the following mixed state:

\begin{align}
    \rho_{ABCE}^{(1)} &= \frac{1}{4} \ketbra{00}{00}_{14} \otimes \big( \ketbra{00}{00}_{23} + \ketbra{11}{11}_{23} \big) \otimes \ketbra{e_0}{e_0}
     \nonumber  \\ 
    &\quad +\frac{1}{4} \ketbra{01}{01}_{14} \otimes \big( \ketbra{00}{00}_{23} + \ketbra{11}{11}_{23} \big) \otimes \ketbra{e_2}{e_2} \nonumber  \\
    &\quad +\frac{1}{4} \ketbra{10}{10}_{14} \otimes \big( \ketbra{00}{00}_{23} + \ketbra{11}{11}_{23} \big) \otimes \ketbra{e_1}{e_1} \nonumber  \\ 
    &\quad +\frac{1}{4} \ketbra{11}{11}_{14} \otimes \big( \ketbra{00}{00}_{23} + \ketbra{11}{11}_{23} \big) \otimes \ketbra{e_3}{e_3}.   
\end{align}

     \item \textbf{Situation 02:} Alice applies the identity, while Bob applies the Hadamard operation. In this situation Charlie will apply $H$ on $s_3^i$ and $I$ on $s_4^i$. Based on (\ref{h on phi+}) and (\ref{h on phi-}), The state shown in Eq. (\ref{v_e on C}) evolves in the following manner:

 \begin{align}
     \frac{1}{2} \begin{bmatrix}
                         ( \ket{00}_{14}\ket{00}_{23} + \ket{00}_{14}\ket{11}_{23})\otimes\ket{e_0} \\
                         +( \ket{01}_{14}\ket{01}_{23} - \ket{01}_{14}\ket{10}_{23})\otimes\ket{e_2} \\
                         +( \ket{10}_{14}\ket{00}_{23} + \ket{10}_{14}\ket{11}_{23})\otimes\ket{e_1} \\
                         +( \ket{11}_{14}\ket{01}_{23} - \ket{11}_{14}\ket{10}_{23})\otimes\ket{e_3}
                 \end{bmatrix}.
\end{align}

Upon measurement by all users, the mixed state of system can be expressed as:

\begin{align}
    \rho_{ABCE}^{(2)} &= \frac{1}{4} \ketbra{00}{00}_{14} \otimes \big( \ketbra{00}{00}_{23} + \ketbra{11}{11}_{23} \big) \otimes \ketbra{e_0}{e_0}
     \nonumber  \\ 
    &\quad +\frac{1}{4} \ketbra{01}{01}_{14} \otimes \big( \ketbra{01}{01}_{23} + \ketbra{10}{10}_{23} \big) \otimes \ketbra{e_2}{e_2} \nonumber  \\
    &\quad +\frac{1}{4} \ketbra{10}{10}_{14} \otimes \big( \ketbra{00}{00}_{23} + \ketbra{11}{11}_{23} \big) \otimes \ketbra{e_1}{e_1} \nonumber  \\ 
    &\quad +\frac{1}{4} \ketbra{11}{11}_{14} \otimes \big( \ketbra{01}{01}_{23} + \ketbra{10}{10}_{23} \big) \otimes \ketbra{e_3}{e_3}.    
\end{align}

\end{enumerate}

Let us define $p_{i,j}$ as the probability that Alice measures the state $\ket{i}$ on qubit $s_1^i$, and Charlie measures $\ket{j}$ on qubit $s_4^i$. It is easy to see that both scenarios yield the same probabilities, given by:

\begin{align} \label{p_00 p_01}
    p_{0,0} &= \frac{\braket{e_0}}{2}, \quad  p_{0,1} = \frac{\braket{e_2}}{2},  \\ \label{p_10 p_11}
    p_{1,0} &= \frac{\braket{e_1}}{2}, \quad  p_{1,1} = \frac{\braket{e_3}}{2}.
\end{align}

Note that those probabilities are statistics that Charlie can observe. Moreover, it is evident that tracing out $B$ and $C$ from both $\rho_{ABCE}^{(1)}$ and $\rho_{ABCE}^{(2)}$ gives rise to the same mixed state shared between Alice and Eve, which is:

\begin{align}
    \rho_{AE} &= \frac{1}{2} \ketbra{0}{0}_1 \otimes (\ketbra{e_0}{e_0} + \ketbra{e_2}{e_2})  \\
              & \quad+ \frac{1}{2} \ketbra{1}{1}_1 \otimes (\ketbra{e_1}{e_1} + \ketbra{e_3}{e_3}).
\end{align}

Therefore, both situations will result in the same key rate.

\subsubsection{Key rate evaluation}

In order to obtain the expression of the key rate in the asymptotic scenario, we first need to calculate $S(A|E)$ and $H(A|C))$. The calculation of the latter is trivial, given (\ref{p_00 p_01}-\ref{p_10 p_11}), we have:

\begin{align} \label{H(A|C)}
    H(A|C) &= H(AC) - H(C), \\
           &= H(p_{0,0}, p_{0,1}, p_{1,0}, p_{1,1})  - H(p_{0,0} + p_{1,0}, p_{0,1} + p_{1,1}).
\end{align}

To lower bound $S(A|E)$, we shall follow the technique outlined in \cite{Krawec2016a}, which is to introduce an auxiliary system $F$ to simplify the calculation. According to the strong sub-additivity of the von Neumann entropy, we have $S(A|E) \geq S(A|EF)$. Thus, our new lower bound for $r$ becomes:

\begin{align} \label{r}
    r \geq S(A|E) - H(A|C) \geq S(A|EF) - H(A|C).
\end{align}

This newly introduced system is spanned by the basis $\{\ket{C}, \ket{W}\}$, where $\ket{C}$ denotes that the measurement results of Alice and Charlie on $s_1^i$ and $s_4^i$, respectively, are correlated, and $\ket{W}$ denotes that their results are anti-correlated. Thus, the new mixed state can be expressed in the following manner:

\begin{align}
    \rho_{AEF} &= \frac{1}{2} \ketbra{0}{0}_1 \otimes \big(\ketbra{e_0}{e_0} \otimes \ketbra{C}{C} + \ketbra{e_2}{e_2} \otimes \ketbra{W}{W}\big)  \\
              & \quad+ \frac{1}{2} \ketbra{1}{1}_1 \otimes \big(\ketbra{e_1}{e_1} \otimes \ketbra{W}{W} + \ketbra{e_3}{e_3} \otimes \ketbra{C}{C}\big).
\end{align}

This state can be written as a diagonal matrix with elements $p_{i,j}$. Therefore:

\begin{align} \label{S(AEF)}
    S(AEF) = H(p_{0,0}, p_{0,1}, p_{1,0}, p_{1,1}).
\end{align}

Next, we need to calculate $S(EF)$. By tracing out the system of Alice, we obtain:

\begin{align}
    \rho_{EF} &= \frac{1}{2} \big( \ketbra{e_0}{e_0} + \ketbra{e_3}{e_3} \big) \otimes \ketbra{C}{C} \\ 
              & \quad +\frac{1}{2} \big( \ketbra{e_1}{e_1} + \ketbra{e_2}{e_2} \big) \otimes \ketbra{W}{W}.
\end{align}

This state can be written in the following manner:

\begin{align} \label{rho_ef}
    \rho_{EF} = \xi_1 \sigma_1 \otimes \ketbra{C}{C} + \xi_2 \sigma_2 \otimes \ketbra{W}{W},
\end{align}
where

\begin{align}
    \xi_1 &= \frac{1}{2} \big( \braket{e_0} + \braket{e_3} \big) = p_{0,0} + p_{1,1}, \\
    \xi_2 &= \frac{1}{2} \big( \braket{e_2} + \braket{e_1} \big) = p_{0,1} + p_{1,0},
\end{align}
and

\begin{align}
    \sigma_1 &= \cfrac{\ketbra{e_0}{e_0} + \ketbra{e_3}{e_3}}{\braket{e_0}{e_0} + \braket{e_3}{e_3}}, \\
    \sigma_2 &= \cfrac{\ketbra{e_2}{e_2} + \ketbra{e_1}{e_1}}{\braket{e_2}{e_2} + \braket{e_1}{e_1}}.
\end{align}

According to the joint entropy theorem, the von Neumann entropy of (\ref{rho_ef}) can be written as:

\begin{align}
    S(EF) = H(p_{0,0} + p_{1,1}, p_{0,1} + p_{1,0}) + (p_{0,0} + p_{1,1})S(\sigma_1) + (p_{0,1} + p_{1,0})S(\sigma_2).
\end{align}

Notice that both $\sigma_1$ and $\sigma_2$ are two-dimensional systems, where $S(\sigma_1) \leq 1$ and $S(\sigma_2) \leq 1$. Therefore, we can upper bound $S(EF)$ in the following way:

\begin{align} \label{S(EF)}
    S(EF) \leq H(p_{0,0} + p_{1,1}, p_{0,1} + p_{1,0}) + (p_{0,0} + p_{1,1})S(\sigma_1) + (p_{0,1} + p_{1,0}).
\end{align}

Now, we need to calculate the eigenvalues of $\sigma_1$. Without loss of generality, we can define 

\begin{align}
    \ket{e_0} &= \alpha \ket{u}, \\
    \ket{e_3} &= \beta \ket{u} + \gamma \ket{v},
\end{align}
where $\alpha, \beta, \gamma \in \mathbb{C}$, $\braket{u} = \braket{v} = 1$, and $\braket{u}{v} = 0$. We also have the following identities:

\begin{align}
    \braket{e_0} &= \abs{\alpha}^2 = 2p_{0,0}, \\
    \braket{e_3} &= \abs{\beta}^2 + \abs{\gamma}^2 = 2p_{1,1}. 
\end{align}

If we express $\sigma_1$ in the basis $\{\ket{u}, \ket{v}\}$, we obtain the following matrix:

\begin{align} 
    \sigma_1 = \frac{1}{2\xi_1} \begin{pmatrix}
        \abs{\alpha}^2 + \abs{\beta}^2 & \beta \gamma^* \\
        \beta^{*}\gamma & \abs{\gamma}^2
     \end{pmatrix}
\end{align}

The eigenvalues of this matrix are:

\begin{align}
    \lambda_{\pm} = \frac{1}{2} \pm \cfrac{\sqrt{(p_{0,0} - p_{1,1})^{2} + \abs{\braket{e_0}{e_3}}^{2}}}{2(p_{0,0} + p_{1,1})}
\end{align}

Thus, we have

\begin{align} \label{sigma_1}
    S(\sigma_1) = -\lambda_+ \log \lambda_+ - \lambda_- \log \lambda_- = h(\lambda_+).
\end{align}

 The final step to complete our proof is to find an upper bound on $h(\lambda_+)$. Since $\lambda_+ \geq \frac{1}{2}$ and $h(\lambda_+)$ decreases as $\lambda_+$ increases (in particular, as $\abs{\braket{e_0}{e_3}}^2$ increases), this reduces to finding a lower bound on $\abs{\braket{e_0}{e_3}}^2$. Let's define $\tilde{\lambda}$ as

 \begin{align} \label{lambda_tilde}
     \tilde{\lambda} = \frac{1}{2} \pm \cfrac{\sqrt{(p_{0,0} - p_{1,1})^{2} + B}}{2(p_{0,0} + p_{1,1})}
 \end{align}
such as $\lambda_+ \geq \tilde{\lambda} \implies h(\lambda_+) \leq h(\tilde{\lambda})$, and where $B$ is the lower bound of $\abs{\braket{e_0}{e_3}}^2$.

To evaluate $B$, we can use the case where Alice and Bob apply $H$ and $I$ respectively. We know that in this situation, Charlie applies the same operation as Alice on $s_4^i$ and the same operation as Bob on $s_3^i$. Therefore, the state (\ref{v_e on C}) evolves and can be written in the following way:

\begin{align} \label{v_e 3rd case}
   \frac{1}{4} \begin{bmatrix}
        \big(\ket{00}_{14}\ket{00}_{23} + \ket{01}_{14}\ket{11}_{23} \big) \ket{f_0} \\
        \big(\ket{01}_{14}\ket{00}_{23} + \ket{00}_{14}\ket{11}_{23} \big) \ket{f_1} \\
        \big(\ket{10}_{14}\ket{00}_{23} + \ket{11}_{14}\ket{11}_{23} \big) \ket{f_2} \\
        \big(\ket{11}_{14}\ket{00}_{23} + \ket{10}_{14}\ket{11}_{23} \big) \ket{f_3}
    \end{bmatrix},
\end{align}
where

\begin{align} \label{f_0}
    \ket{f_0} &= \ket{e_0} + \ket{e_1} + \ket{e_2} + \ket{e_3}, \\
    \ket{f_1} &= \ket{e_0} + \ket{e_1} - \ket{e_2} - \ket{e_3}, \\
    \ket{f_2} &= \ket{e_0} - \ket{e_1} + \ket{e_2} - \ket{e_3}, \\ \label{f_3}
    \ket{f_3} &= \ket{e_0} - \ket{e_1} - \ket{e_2} + \ket{e_3}.
\end{align}

Charlie and the participants measure their respective qubits in the $Z$ basis. Let $p_c$ denote the probability that their measurement outcomes are consistent with the expected results shown in Table \ref{table 1}, and let $p_{in}$ denote the probability that the outcomes are inconsistent with those expected results. From (\ref{v_e 3rd case}), and (\ref{f_0}-\ref{f_3}), we can easily get:

\begin{align} \label{p_c}
    p_c &= \frac{1}{8} \big( \braket{f_0} + \braket{f_3} \big), \nonumber \\
        &= \frac{1}{2} \big( 1 + Re\braket{e_0}{e_3} + Re\braket{e_1}{e_2}  \big). 
\end{align}

\begin{align} \label{p_in}
    p_{in} &= \frac{1}{8} \big( \braket{f_1} + \braket{f_2} \big), \nonumber \\
        &= \frac{1}{2} \big( 1 - Re\braket{e_0}{e_3} - Re\braket{e_1}{e_2}  \big). 
\end{align}

By subtracting those two probabilities, we obtain:

\begin{align}
    Re\braket{e_0}{e_3} = (p_c - p_{in}) - Re\braket{e_1}{e_2}.
\end{align}

By using the Cauchy-Schwarz inequality $\abs{Re\braket{x}{y}} \leq \abs{\braket{x}{y}} \leq \sqrt{\braket{x}{x}\braket{y}}$, we can obtain a lower bound for the above quantity, which is:

\begin{align}
    Re\braket{e_0}{e_3} \geq (p_c - p_{in}) - 2\sqrt{p_{0,1}p_{1,0}}.
\end{align}

Since $\abs{\braket{e_0}{e_3}} \geq \abs{Re\braket{e_0}{e_3}}$, our final bound can be written as:

\begin{align} \label{bound}
    \abs{\braket{e_0}{e_3}}^2 \geq \abs{(p_c - p_{in}) - 2\sqrt{p_{0,1}p_{1,0}}}^2 := B.
\end{align}

By replacing the expression of (\ref{H(A|C)}), (\ref{S(AEF)}), and (\ref{S(EF)}) in (\ref{r}), the key rate can be expressed as follows:

\begin{align} \label{r_final}
    r &\geq H(p_{0,0}+p_{1,0}, p_{0,1}+p_{1,1}) - H(p_{0,0}+p_{1,1}, p_{0,1}+p_{1,0}) \nonumber \\
    & \quad- (p_{0,0}+p_{1,1})h(\tilde{\lambda}) - (p_{0,1}+p_{1,0})
\end{align}

Notice that all the variables inside $r$ are statistics that Charlie can observe. Therefore, the value of the key rate can be derived. 

It is a common assumption in QKD literature to model Eve's attack as a depolarization channel with parameter $Q$. This parameter represents the error rate introduced by Eve's attack and can be interpreted as the probability that Alice and Charlie obtain anti-correlated results on their measurements on $s_1^i$ and $s_4^i$, respectively. We also emphasis that $Q$ can easily be observed by Charlie. In this scenario, our probabilities (\ref{p_00 p_01}), (\ref{p_10 p_11}), (\ref{p_c}) and (\ref{p_in}) can be expressed as follows:

\begin{align} \label{p_ij Q}
    p_{0,0} = p_{1,1} = \frac{(1-Q)}{2}, \quad\quad p_{0,1} = p_{1,0} = \frac{Q}{2}.
\end{align}

\begin{align} \label{p_c_in Q}
    p_c &= (1-Q), \quad\quad p_{in} = Q.
\end{align}

By substituting the above expressions into Eq. (\ref{r_final}), the key rate can be expressed solely as a function of the error rate $Q$. As shown in Figure \ref{key_rate_fig}, Charlie and Alice can distill a secure shared key as long as $Q \leq 9.68\%$. This value represents the noise tolerance of our protocol, which is close to the $11\%$ threshold of the BB84 protocol.

\begin{figure}[h]
    \centering
    \includegraphics[width=0.7\textwidth]{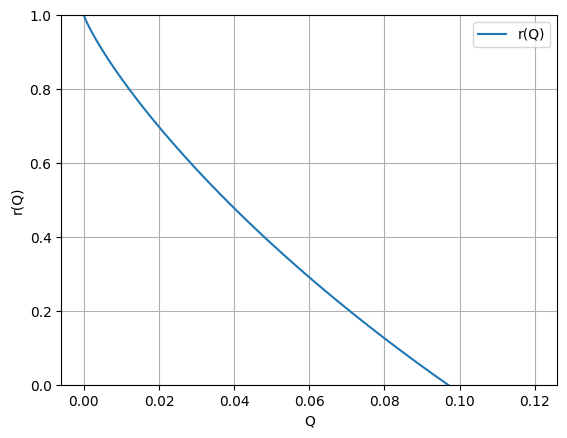}  
    \caption{Key rate of the proposed scheme as a function of $Q$.}
    \label{key_rate_fig}
\end{figure}

\section{Discussion and comparison}\label{sec5_comp}

In this section, we compare the proposed protocol with some existing multi-party SQKD schemes \cite{XianZhou2009,Zhou2019,Tian2022,Ye2023,Tsai2024} in the three-party scenario. The comparison is drawn from six perspectives: quantum resources, quantum capabilities of the participants, communication structure, whether the protocol can mitigate Trojan horse attacks without additional costly devices, noise tolerance, and qubit efficiency. The latter is calculated using the following formula:

\begin{align}
    \eta = \frac{c}{(q+b)},
\end{align}
where $c$ is the length of the raw key, $q$ is the total number of qubits generated by the quantum user, and $b$ is the total number of qubits generated by the classical participants. Since the proposed protocol is a one-way communication scheme, the classical participants do not need to generate any qubits to send back to Charlie. Thus, $b=0$. Under the ideal situation, that is $\epsilon = 0$, Charlie generates $4n$ four-particle cluster states, which means $q = 16n$. At the end of the protocol, both Alice and Bob obtain a raw key $R_{CA}$ and $R_{CB}$, respectively, each of length $n$. Therefore, the qubit efficiency of our protocol is $\eta = 1/8$.

In terms of quantum resources, Zhang et al.'s protocol uses single qubits, making it more practical to implement compared to other protocols that rely on entangled states. However, this scheme suffers from poor qubit efficiency. The same issue applies to other protocols employing a circular communication structure (i.e., those by Tian et al. and Ye et al.). This is due to fact that all participants must  choose the “correct” operation (i.e. measure-and-resend) in order to establish a shared secret key. Tsai et al.'s protocol also exhibits low efficiency, primarily due to the use of quantum repeaters, which negatively impact qubit efficiency. In contrast, our proposed protocol achieves the highest qubit efficiency among all the schemes listed in Table \ref{table 2}.

Regarding quantum capabilities, the proposed protocol and that of Tsai et al. are more lightweight than existing protocols, as classical users require only two quantum capabilities, compared to the three needed in other schemes. Additionally, the adoption of a one-way communication structure offers several advantages. First, the proposed protocol benefits from shorter qubit transmission distances than protocols employing circular or two-way communication, resulting in a lower cost for qubit transmission. Second, other protocols must incorporate additional, costly devices or mechanisms to defend against quantum Trojan horse attacks. This not only increases the quantum burden on classical participants and thus contradicting the very premise of the semi-quantum model, but also reduces qubit efficiency. In contrast, the proposed protocol is inherently immune to such attacks, as classical participants are not required to use any additional devices.

In terms of noise tolerance, our protocol achieves a threshold of $9.68\%$, the highest among the values listed in Table \ref{table 2}, closely followed by Ye et al.'s protocol. However, due to the absence of necessary parameters in some existing protocols, a complete quantitative comparison is not feasible.

Overall, the advantages of our protocol over existing schemes can be summarized as follows:

\begin{enumerate}
    \item Classical participants require only two quantum capabilities, rather than three.
    \item The protocol is inherently robust against Trojan horse attacks, without the need for additional costly devices or mechanisms on the part of the classical participants.
    \item It benefits from a shorter qubit transmission distance compared to certain other schemes.
    \item The protocol exhibits superior qubit efficiency.
    \item It offers higher noise tolerance than several existing protocols.
\end{enumerate}

Two important observations should be made before concluding the comparison. First, among all protocols considered, only our scheme and that of Zhou et al. allow the quantum user to share a distinct secret key with each classical participant. In contrast, the other protocols are designed to establish a common secret key shared by all participants. Second, the comparison presented here focuses on the three-party scenario. Extending the proposed protocol to accommodate $m$ participants using a $2m$-particle cluster state is an issue left for future work.

\begin{table}[h]
\caption{\label{table 2}Comparison of proposed protocol with other MSQKD protocols.}
\resizebox{\textwidth}{!}{%
\begin{tabular}{lcccccc}
\toprule
\textbf{} & \makecell{Quantum \\ resources} & \makecell{Capabilities of \\ classical participants} & \makecell{Communication \\ structure}  & \makecell{Mitigate THA\\ without device} & \makecell{Key \\ rate} & \makecell{Qubit \\ efficiency} \\
\midrule
Zhang et al. \cite{XianZhou2009} & Single photons & \multirow{4}{*}{\makecell[l]{1. Generate $\{\ket{0}, \ket{1}\}$\\2. Measure $\{\ket{0}, \ket{1}\}$\\3. Reflect}} & Circular & No & -- & $\cfrac{1}{32}$ \\

Zhou et al. \cite{Zhou2019} & Cluster states &  & Two-way & No & $2.82\%$ & $\cfrac{1}{10}$ \\

Tian et al. \cite{Tian2022} & \makecell{Hyperentangled\\Bell states} &  & Circular & No & -- & $\cfrac{1}{24}$ \\

Ye et al. \cite{Ye2023} & Bell states &  & Circular & No & $9.53\%$ & $\cfrac{1}{24}$ \\

Tsai et al. \cite{Tsai2024} & Graph state & \makecell[l]{1. Perform Pauli and $H$ operations\\2. Measure $\{\ket{0}, \ket{1}\}$} & One-way & Yes & -- & $\cfrac{1}{8(3 + r^*)}$ \\

Our scheme & Cluster states & \makecell[l]{1. Perform $H$\\2. Measure $\{\ket{0}, \ket{1}\}$} & One-way & Yes & $9.68\%$ & $\cfrac{1}{8}$ \\
\bottomrule
\end{tabular}
}
\noindent\footnotesize
($*$) $r$ is the number of quantum repeaters.
\end{table}

\section{Conclusion} \label{sec6_conc}

This paper introduced a novel lightweight quantum key distribution (LQKD) protocol that enables a quantum-capable user, Charlie, to simultaneously establish two separate secret keys, one with classical Alice and another with classical Bob. This multi-user setting offers better time and qubit efficiency compared to traditional two-party schemes. The protocol leverages the four-particle cluster state, whose properties make it a more suitable quantum resource for multi-party QKD protocols than other entangled states, such as GHZ-states.

By adopting a one-way qubit transmission method, the protocol effectively addresses several challenges in existing SQKD schemes:
(1) it significantly reduces the quantum overhead for classical users, who require only two quantum capabilities (applying the Hadamard operation and performing $Z$-basis measurements), without the need for additional costly devices to counter quantum Trojan horse attacks; (2) it benefits from shorter qubit transmission distances than protocols based on circular or two-way communication, thus lowering the cost of quantum transmission; and (3) it achieves higher time and qubit efficiency. As a result, the proposed protocol is both more lightweight and more practical than existing SQKD protocols. 

The protocol leverages the relationship between Bell states and the Hadamard operation to detect eavesdroppers, and the security analysis shows that the schemes is secure against both internal and external adversaries in the ideal setting. Furthermore, this study provides a detailed proof of the protocol’s unconditional security. Specifically, a lower bound on the key rate and noise tolerance is derived in the asymptotic scenario, based on observable parameters in the channel. It is shown that users can extract a secure secret key as long as the noise remains below $9.68\%$, a threshold close to the BB84 protocol’s $11\%$.

In our future work, it would be interesting to extend the protocol to multi-user scenarios and to derive the key rate expression in the finite-key setting.

\bibliographystyle{plainnat}
 \bibliography{cluster_ref}

\end{document}